\documentclass[reqno]{amsart}

\usepackage{enumitem}
\usepackage{hyperref}
\usepackage[all]{xy}
\usepackage{graphicx}
\usepackage{subfigure}
\usepackage{amsthm, amssymb, amsmath,mathrsfs}
\usepackage[utf8]{inputenc}

\newtheorem{theorem}{Theorem}[section]

\newtheorem{lemma}{Lemma}[section]

\newcommand{\be}{\begin{equation}}
\newcommand{\ee}{\end{equation}}

\newcommand{\E}{\mathrm{e}}
\newcommand{\I}{\mathrm{i}}
\newcommand{\re}{\mathrm{Re}}
\newcommand{\im}{\mathrm{Im}}

\newcommand*{\mailto}[1]{\href{mailto:#1}{\nolinkurl{#1}}}

\newcommand{\doi}[1]{\href{http://dx.doi.org/#1}{doi: #1}}
\newcommand{\msc}[1]{\href{http://www.ams.org/msc/msc2010.html?t=&s=#1}{#1}}

\numberwithin{equation}{section}

\begin{document}

\title{A General Uncertainty Principle for  Partial Differential Equations}

\author[I. Alvarez-Romero]{Isaac Alvarez-Romero}
\address{Faculty of Mathematics\\University of Vienna\\Oskar-Morgenstern-Platz 1\\1090 Wien\\Austria}
\email{\mailto{isaac.alvarez.romero@univie.ac.at }\\
\mailto{isaacalrom@gmail.com}}

\thanks{{\it Research supported by the FWF Der Wissenschaftsfonds Project M-2335-N32.}\\ Journal of Mathematical Analysis and Applications (to appear)}


\begin{abstract}
We consider the coupled equations
\begin{equation*}
\begin{pmatrix}r_t\\ -q_t\end{pmatrix}+2A_0(L^+)\begin{pmatrix}r\\ q\end{pmatrix}=0,
\end{equation*}
where $L^+$  is the integro-differential operator
\begin{equation*}
L^+=\frac{1}{2\I}\begin{pmatrix}\partial_x-2r\int_{-\infty}^xdyq& 2r\int_{-\infty}^xdyr\\ 
-2q\int_{-\infty}^xdyq& -\partial_x+2q\int_{-\infty}^xdyr.\end{pmatrix}
\end{equation*}
and $A_0(z)$ is an arbitratry ratio of entire functions. We study two main cases: the first one when the potentials $|q|,|r|\to 0$ as $|x|\to\infty$ and the second one when $r=-1$ and $|q|\to0$ as $|x|\to\infty$. In such conditions we prove that there cannot exist a solution different from zero if at two different times the potentials have a strong decay. This  decay is of exponential rate: $\exp(-x^{1+\delta})$, $x\geq 0$ and $\delta>0$ is a constant. As particular cases we will cover the Korteweg-de Vries equation, the modified Korteweg-de Vries equation and the nonlinear Schr\"odinger equation.
\end{abstract}

\maketitle

\section{Introduction}

In this paper we consider the coupled equations
\begin{equation}\label{NLSq:10}
\begin{pmatrix}r_t\\ -q_t\end{pmatrix}+2A_0(L^+)\begin{pmatrix}r\\ q\end{pmatrix}=0,
\end{equation}
where $L^+$  is the integro-differential operator
\begin{equation*}
L^+=\frac{1}{2\I}\begin{pmatrix}\partial_x-2r\int_{-\infty}^xdyq& 2r\int_{-\infty}^xdyr\\ 
-2q\int_{-\infty}^xdyq& -\partial_x+2q\int_{-\infty}^xdyr\end{pmatrix}
\end{equation*}
and $A_0(z)$ is an arbitrary ratio of entire functions, which is directly related with the dispersion relation of the linearized version of \eqref{NLSq:10}, and the potentials $(r,q)$ also fulfill some decay condition as $|x|\to \infty$.

We prove that a solution of \eqref{NLSq:10} cannot decay faster than an exponential rate at two different times,  $t_0,t_1$, unless it is trivial, that is, if  $|r(x,t_i)|\leq C_1 \E^{-C_2 x^{1+\delta}}$, $|q(x,t_i)|\leq C_3 \E^{-C_4 x^{1+\beta}}$, for $x\geq 0$, $t_i=t_0,t_1$ and $C_i>0$, $i=1,2,3,4$,  $\delta, \beta$ positive constants, then $r,q\equiv 0$ for all $x\in\mathbb{R}$ and any time $t$. How small these two constants $\alpha,\beta$ can be, depends directly on $A_0(z)$. For instance, if $A_0(z)=-4\I z^3$ and $r=-1$, then \eqref{NLSq:10} turns into the Korteweg-de Vries equation and we will prove that $\delta >1/2$, and if $A_0(z)=-2\I z^2$ and $r=-q^*$, where '$*$' denotes the complex conjugate, then \eqref{NLSq:10} turns into the nonlinear Schr\"odinger equation and $\delta>1$. Similarly, one can observe that if $r=q$ and $A_0(z)=-4\I z^3$, then we obtain the modified Korteweg-de Vries equation, and $\delta>1/2$.

In recent years there has been a reinterpretation of the uncertainty principle given by G. H. Hardy in \cite{H} as a dynamic version for the free Schr\"odinger equation, see \cite{CEKPV} and the references therein, that is, if $u(t,x)$ is a solution for $\partial_t u=\I\Delta u$ and $|u(0,x)|=O(\E^{-x^2/\beta^2})$, $|u(1,x)|=O(\E^{-x^2/\alpha^2})$, with $1/\alpha\beta>1/4$, then $u\equiv 0$ and if $1/\alpha\beta=1/4$, then the initial data is a multiple of $\E^{-(1/\beta^2+\I/4)x^2}$.  To prove this result one can use techniques of real analysis, in particular Carleman estimates and properties of the log-convexity of some special functions.  In addition, the authors L. Escauriaza, C. E. Kenig, G. Ponce and L. Vega  also manage to use similar arguments to prove an uncertainty principle for nonlinear problems, such as the Korteweg-de Vries equation in \cite{EKPV} and the nonlinear Schr\"odinger equation in \cite{EKPV1} as one can see in the following theorems:

\begin{theorem}(EKPV)\label{th001}
Let $u_1,u_2\in C([0,1]:H^3(\mathbb{R})\cap L^2(|x|^2dx))$ be strong solutions of 
\begin{equation*}
\partial_tu +\partial_x^3u+u^k\partial_xu=0,\quad\quad(x,t)\in\mathbb{R}^2, k\in\mathbb{Z}^+
\end{equation*}
in the domain $(x,t)\in\mathbb{R}\times[0,1]$. If
\begin{equation*}
u_1(\cdot,0)-u_2(\cdot,0),\quad u_1(\cdot,1)-u_2(\cdot,1)\in H^1\Big(\E^{ax_+^{3/2}}dx\Big),
\end{equation*}
for any $a>0$, then $u_1\equiv u_2$.
\end{theorem}
They denote  $f\in H^1\Big(\E^{ax_+^{3/2}}dx\Big)$ if $f,\partial_x f\in L^2(\E^{ax_+^{3/2}}dx)$, where $x_+=\max\{0,x\}$.

\begin{theorem}(EKPV)\label{th002}
Let $u_1,u_2\in C([0,1]: H^k(\mathbb{R}^n))$, $k\in\mathbb{Z}^+$, $k>n/2+1$ be strong solutions of the equation 
\begin{equation*}
\I\partial_t u+\Delta u+F(u,u^*)=0
\end{equation*}
in the domain $(x,t)\in\mathbb{R}^n\times[0,1]$, with $F:\mathbb{C}^2\to\mathbb{C}$, $F\in C^k$ and $F(0)=\partial_u F(0)=\partial_{u^*}F(0)=0$.
If there exist $\alpha>2$ and $a>0$ such that
\begin{equation*}
w_0=u_1(\cdot,0)-u_2(\cdot,0),\quad w_1=u_1(\cdot,1)-u_2(\cdot,1)\in H^1(\E^{a|x|^\alpha}dx),
\end{equation*}
then $u_1\equiv u_2.$
\end{theorem}
Theorem \ref{th001} corresponds to the generalized Korteweg-de Vries equation and theorem \ref{th002} corresponds to the nonlinear Schr\"odinger equation.

Combining complex analysis and scattering theory we obtain another approach to prove  uncertainty principles. This  was shown  for instance in \cite{JLMP} for the discrete Schr\"odinger equation. They also proved something a little stronger 
\begin{theorem}(JLMP)\label{thJLMP}
Let $u(t,n)$, $t>0$, $n\in\mathbb{Z}$ be a strong solution of 
\begin{equation*}
\partial_tu=\I(\Delta_du+Vu),
\end{equation*}
where $\Delta_d$ is the discrete Laplacian: $\Delta_df(n):=f(n+1)+f(n-1)-2f(n)$. Assume  that the potential $V$ does not depend on time and also $V(n)\neq 0$ just for a finite numbers of $n$'s. If, for some $\epsilon>0$,
\begin{equation*}
|u(t,n)|\leq C\Big(\frac{\E}{(2+\epsilon)n}\Big)^n,\quad n>0, t\in\{0,1\},
\end{equation*}
then $u\equiv0$.
\end{theorem}

In case that the potential $V(n)$ depends on $t$,  one has to use techniques of real analysis, as it was done in \cite{CEKPV}. This was also proved in \cite{JLMP}.

Using a similar method of combining the scattering theory and complex analysis, I. Alvarez-Romero and G. Teschl proved in \cite{ART1}
the extension of the theorem  \ref{thJLMP} to general Jacobi operators with exponential decay rate. As a natural continuation, in \cite{ART}, the authors applied a similar technique to a nonlinear problem:  the Toda hierarchy. Here they proved that if a solution for the Toda hierarchy  has a strong decay it must be trivial.  The pattern here is to show that the reflection coeffcient of the scattering data must vanish and then one proves that there is no pure N-soliton solution for the Toda system with such strong decay condition.

Following this line of work, in this paper we prove an uncertainty principle in its dynamic form for partial differential equations which are nonlinear.  Classical analysis for nonlinear partial differential  equations uses harmonic analysis and the dispersion relation of the propagation waves. Actually, using numerical methods, one can obtain a solution in terms of the initial data, but in general little more can be said about the exact solution. Thus, in order to work in a general setting, and not only the Korteweg-de Vries or nonlinear Schr\"odinger equation, we will need another point of view to study how the solutions of some partial differential equations behave.   We will need to use the inverse scattering transform on these equations. As it is well known, in 1967 in \cite{GGKM}, C. S. Gardner, J. M. Greene, M. D. Kruskal and  R. M. Miura discovered the method to apply the inverse scattering transform to the Korteweg-de Vries equation
\begin{equation*}
u_t+uu_x+u_{xxx}=0
\end{equation*}
when the initial data $u(0,x)$ is given and $u(0,x)\to 0$ sufficiently rapidly as $|x|\to\infty$.  And later, in 1972, V. F. Zakharov and A. B. Shabat in \cite{ZS}
used the scattering problem
\begin{equation*}
\begin{split}
&v_{1x}+\I\lambda v_1=q(x,t)v_2\\
&v_{2x}-\I\lambda v_2=r(x,t)v_1, \quad -\infty<x<\infty
\end{split}
\end{equation*}
with $r=-q^*$, where '$*$' denotes the complex conjugation, to find the solution for the nonlinear Schr\"odinger equation:
\begin{equation*}
q_t-\I q_{xx}-2\I q^2q^*=0.
\end{equation*}

This second  example of the use of the inverse scattering transform was strong evidence that the method was not fortuitous. Thus, in  1974, M. J. Ablowitz, D. J. Kaup,  A. C. Newell and   H. Segur in \cite{AKNS} gave a detailed description of a wider setting of partial differential equations where this method could be applied. They considered the coupled equations
\begin{equation}\label{NLSq:1}
\begin{pmatrix}r_t\\ -q_t\end{pmatrix}+2A_0(L^+)\begin{pmatrix}r\\ q\end{pmatrix}=0,
\end{equation}
where $L^+$  is the integro-differential operator
\begin{equation*}
L^+=\frac{1}{2\I}\begin{pmatrix}\partial_x-2r\int_{-\infty}^xdyq& 2r\int_{-\infty}^xdyr\\ 
-2q\int_{-\infty}^xdyq& -\partial_x+2q\int_{-\infty}^xdyr\end{pmatrix}
\end{equation*}
and $A_0(z)$ is an arbitrary ratio of entire functions, which is directly related with the dispersion relation of the linearized version of \eqref{NLSq:1}. In fact, as we will see in some examples, this function $A_0(z)$ is of the order of the dispersion relation. As it happened in the Korteweg-de Vires equation, the potentials $(r,q)$ also fulfill some decay condition as $|x|\to \infty$.

Thus, in the present paper we obtain a general uncertainty principle for a general partial differential equation of the form \eqref{NLSq:1}. And as it happened in \cite{ART1,ART,JLMP}, we will combine the scattering theory from \cite{AKNS}
and complex analysis, in particular the study of the growth of entire functions.
As particular cases, we will observe that for the special cases of the Korteweg-de Vries equation we will obtain the same bound on the decay condition, that is $\E^{-x^{3/2}}$, as in theorem \ref{th001} and the same will happen for the nonlinear Schr\"odinger equation in theorem \ref{th002}, that is, the decay is of the form $\E^{-x^2}$. This happens when one of the two solutions, which we are comparing, is the zero one.

\section{Preliminaries}

\subsection{Growth of entire functions.} We will give some results, which can be found all of them in \cite{Levin} in lectures 1 and 8, about the asymptotic behaviour of entire functions.

 Let $f(z)$ be an entire function and $M_f(r)=\max_{|z|=r}|f(z)|$, we say that $f(z)$ is of {\em finite order} if for some $k\geq 0$ and $r$ big enough,
\begin{equation}\label{Eq:1}
M_f(r)<\exp(r^k)
\end{equation}
 and we say that $f$ has {\em finite type}, $\sigma_f$, if for $|z|$ big enough and some $\sigma>0$ we have 
\begin{equation*}
|f(z)|<\exp(\sigma|z|^\rho),
\end{equation*}
where  the {\em order } of $f$, $\rho$, will be the greatest lower bound of the $k$ which fulfills \eqref{Eq:1}, similarly, we define the  {\em type} of $f$, $\sigma_f$, with respect to the order $\rho$ . This notion of $\sigma_f$ gives us an idea of how fast the function $f(z)$ grows as $|z|\to\infty$. However, it may happen that $f(z)$ grows faster in one direction than another, thus we define the {\em indicator function} of $f$ with respect to the order $\rho$ by
\begin{equation*}
h_f(\varphi)=\limsup_{r\to\infty}\frac{\log|f(r\E^{\I \varphi})|}{r^\rho},\quad \varphi\in[0,2\pi].
\end{equation*}
It follows from the definition
\begin{equation}\label{Eq:3}
\begin{split}
h_{fg}(\varphi)&\leq h_f(\varphi)+h_g(\varphi)\\
h_{f+g}(\varphi)&\leq \max\{h_f(\varphi),h_g(\varphi)\},
\end{split}
\end{equation}
where $f,g$ are two entire functions. Moreover  if $f$ is a function which is analytic inside an angle $D=\{z=r\E^{\I \varphi}: \alpha<\varphi<\beta\}$ and we have the asymptotic inequality as $r\to\infty$: $M_f(r)<\exp(Ar^\rho)$, then
\begin{equation}\label{Eq:4}
h_f(\varphi)+h_f(\varphi+\pi/\rho)\geq 0, \quad \alpha\leq\varphi<\varphi+\pi/\rho\leq\beta.
\end{equation}
Since the definitions of  order, type and indicator function of $f$ deal with the asymptotic behaviour of $f$, they can be extended to functions  which are analytic in a region $\{z:|z|>c>0\}$. In addition due to the fact that  the key ingredient to prove \eqref{Eq:4} is the Phragm\'en-Lindel\"of theorem, it is easy to adapt the proof to this wider set of  functions.

\subsection{Inverse scattering transform}
 
Here we show a brief description of how the inverse scattering transform can be applied to a wider set of partial differential equations, besides the Korteweg-de Vries and nonlinear Schr\"odinger equations, in order  to solve them. A detailed description can be found in  \cite{AKNS}. We also name some results of the Schr\"odinger equation on the real line, which can be found in \cite{DT}. We repeat them here in order to make our exposition self-contained.

Consider the coupled equations
\begin{equation}\label{NLSEq:1}
\begin{pmatrix}r_t\\ -q_t\end{pmatrix}+2A_0(L^+)\begin{pmatrix}r\\ q\end{pmatrix}=0,
\end{equation}
where $L^+$  is the integro-differential operator
\begin{equation*}
L^+=\frac{1}{2\I}\begin{pmatrix}\partial_x-2r\int_{-\infty}^xdyq& 2r\int_{-\infty}^xdyr\\ 
-2q\int_{-\infty}^xdyq& -\partial_x+2q\int_{-\infty}^xdyr\end{pmatrix}
\end{equation*}
and $A_0(z)$ is an arbitrary ratio of entire functions, which is directly related with the dispersion relation of the linearized version of \eqref{NLSEq:1}. For instance if we set $r=-q^*$, where '$^*$' denotes the complex conjugated, and $A_0(z)=-2\I z^2$, \eqref{NLSEq:1} becomes the nonlinear Schr\"odinger equation.

In order to solve \eqref{NLSEq:1}, we want to apply the inverse scattering transform, thus we need the associated eigenvalue problem:
\begin{equation}\label{NLSEq:2}
\begin{split}
v_{1x}+\I\lambda v_1&=q(x,t)v_2\\
v_{2x}-\I\lambda v_2&=r(x,t)v_1,\quad\quad -\infty<x<\infty.
\end{split}
\end{equation}
Here $\lambda $ will play the role of the eigenvalue and the potentials $q,r$ fulfill the following evolution equation:
\begin{equation}\label{NLSEq:3}
\begin{split}
v_{1t}&=A(x,t,\lambda)v_1+B(x,t,\lambda)v_2\\
v_{2t}&=C(x,t,\lambda)v_1+D(x,t,\lambda).
\end{split}
\end{equation}
From cross differentiation of \eqref{NLSEq:2} and \eqref{NLSEq:3}, we obtain $D=-A+d(t)$ and without lost of generality, we assume that $d(t)=0$. We notice that the eigenvalue $\lambda$ is assumed to be independent of time. Thus
\begin{equation}\label{NLSEq:4}
\begin{split}
A_x&=qC-rB\\
B_x+2\I\lambda B&=q_t-2Aq\\
C_x-2\I\lambda C&=r_t+2Ar.
\end{split}
\end{equation}

We will study two different cases:
\begin{itemize}
\item[i)] $q,r\to 0$, as $|x|\to\infty$. The most known example is the nonlinear Schr\"odinger equation (NLS case).
\item[ii)]$r=-1$, which will transform \eqref{NLSEq:2} into the Schr\"odinger equation and as a main example we will have the Korteweg-de Vries equation. A detailed description is given in subsection \ref{SecKdV} (KdV case).
\end{itemize}

\subsubsection{NLS case}
As we have pointed out, we assume that the potentials $q,r$ have some decay as $|x|\to\infty$, i.e. $|r|,|q|\to0$ as $|x|\to\infty$. This decay, from now on, will be 
\begin{equation*}
\begin{split}
&\int_{\mathbb{R}}(1+|x|)|q(x,t)|dx<\infty\\
&\int_{\mathbb{R}}(1+|x|)|r(x,t)|dx<\infty.
\end{split}
\end{equation*}
Thus, we can define the Jost solutions, which (for $\lambda\in\mathbb{R}$) have the following asymptotic values:
\begin{equation}\label{NLSEq:5}
\begin{split}
&\phi\to\begin{pmatrix}1\\0\end{pmatrix}\E^{-\I\lambda x}\quad\quad x\to-\infty\\
&\overline{\phi}\to\begin{pmatrix}0\\-1\end{pmatrix}\E^{\I\lambda x}\quad\quad x\to-\infty\\
&\psi\to\begin{pmatrix}0\\1\end{pmatrix}\E^{\I\lambda x}\quad\quad x\to\infty\\
&\overline{\psi}\to\begin{pmatrix}1\\0\end{pmatrix}\E^{-\I\lambda x}\quad\quad x\to\infty.
\end{split}
\end{equation}
The scattering data, $a(\lambda,t),b(\lambda,t),\overline{a}(\lambda,t)$ and $\overline{b}(\lambda,t)$, appears naturally to relate these solutions:
\begin{equation}\label{NLSEq:6}
\begin{split}
&\phi=a\overline{\psi}+b\psi\to\begin{pmatrix}a\E^{-\I\lambda x}\\ b\E^{\I \lambda x}\end{pmatrix} \quad \quad\text{as }  x\to\infty\\
&\overline{\phi}=\overline{b}\text{ } \overline{\psi}-\overline{a}\psi\to\begin{pmatrix}\overline{b}\E^{-\I \lambda x}\\-\overline{a}\E^{\I\lambda x}\end{pmatrix}\quad\quad\text{as }x\to\infty.
\end{split}
\end{equation}
The coefficients $a(\lambda,t),b(\lambda,t),\overline{a}(\lambda,t)$ and $\overline{b}(\lambda,t)$ are given by the Wronskian of the Jost solutions:
\begin{equation}\label{NLSEq:7}
\begin{split}
&a=W(\phi,\psi)\\
&b=-W(\phi,\overline{\psi})\\
&\overline{a}=W(\overline{\phi},\overline{\psi })\\
&\overline{b}=W(\overline{\phi},\psi),
\end{split}
\end{equation}
where the Wronskian is defined by $W(u,v)=u_1v_2-u_2v_1$. Moreover, $a(\lambda,t)$ can be analytically extended into the upper half plane, $\im(\lambda)>0$ and, respectively, $\overline{a}(\lambda,t)$ into the lower half plane, $\im(\lambda)<0$. In addition, the discrete eigenvalues, also called {\em bound states}, $\{\lambda_k\}_{k=1}^N$ of \eqref{NLSEq:2} are given by the zeros of $a(\lambda,t)$, at which $\phi(\lambda_k,t)=b_k(t)\psi(\lambda_k,t)$,  and similarly the zeros of $\overline{a}(\lambda,t)$ in the lower half plane are also eigenvalues and $\overline{\phi}(\overline{\lambda}_k,t)=\overline{b}_k(t)\overline{\psi}_k(\overline{\lambda}_k,t)$. In general, $\lambda_k,\overline{\lambda}_k$ are not related. We are going to assume that they are a finite set and for simplicity that $\im(\lambda_k)>0$ and $\im(\overline{\lambda}_k)<0$.

Because of the choice of the normalization in \eqref{NLSEq:5}, one can assume without lost of generality that $B,C$ tend to zero as $x\to-\infty$, and also from \eqref{NLSEq:4} that $A(x,t,\lambda)$ tends to a constant, i.e.
\begin{equation*}
\lim_{x\to-\infty}A(x,t,\lambda)=A_-(\lambda),
\end{equation*}
where $A_-(\lambda)$ is an arbitrary function of $\lambda$. In addition, using \eqref{NLSEq:2} and \eqref{NLSEq:3} we observe that
\begin{equation}\label{NLSEq:9}
\begin{split}
\phi_t&=\begin{pmatrix}A-A_-& B\\ C& -A-A_-\end{pmatrix}\phi\\
\overline{\phi}_t&=\begin{pmatrix}A+A_-& B\\ C& -A+A_-\end{pmatrix}\overline{\phi}.
\end{split}
\end{equation}
Hence, using \eqref{NLSEq:6} and \eqref{NLSEq:9} we obtain the time evolution for the scattering data:
\begin{equation}\label{NLSEq:10}
\begin{split}
a_t&=(A_+-A_-)a+B_+b\\
b_t&=C_+a-(A_++A_-)b\\
\overline{a}_t&=-(A_+-A_-)\overline{a}-C_+\overline{b}\\
\overline{b}_t&=-B_+\overline{a}+(A_++A_-)\overline{b},
\end{split}
\end{equation}
where
\begin{equation*}
A_+=\lim_{x\to\infty}A,\quad B_+=\lim_{x\to\infty}B\E^{2\I\lambda x},\quad\text{and }\quad C_+=C\E^{-2\I\lambda x}.
\end{equation*}

In what follows we will assume that $A_+=A_-$ and $B_+=C_+=0$. Thus, the evolution equations for the scattering data turn into
\begin{equation}\label{NLSEq:11}
\begin{split}
a(\lambda,t)&=a(\lambda,0)\\
b(\lambda,t)&=b(\lambda,0)\E^{-2A_-(\lambda)t}\\
\overline{a}(\lambda,t)&=\overline{a}(\lambda,0)\\
\overline{b}(\lambda,t)&=\overline{b}(\lambda,0)\E^{2A_-(\lambda)t}
\end{split}
\end{equation}
and the function  $A_0(\lambda)$ in\eqref{NLSEq:1} is equal to $A_-(\lambda)$, i.e.,  $A_0(\lambda)=A_-(\lambda)$.

Notice that from \eqref{NLSEq:11}, it follows that the eigenvalues of \eqref{NLSEq:2} are the same for all $t$.

The Jost solutions admit a  representation 
\begin{equation}\label{NLSEq:001}
\begin{split}
\psi(\lambda,x)&=\begin{pmatrix}0\\1\end{pmatrix}\E^{\I\lambda x}+\int_{x}^\infty K(x,s)\E^{\I\lambda s}ds\\
\overline{\psi}(x)&=\begin{pmatrix}1\\0\end{pmatrix}\E^{-\I\lambda x}+\int_{x}^\infty \overline{K}(x,s)\E^{-\I\lambda s}ds\\
\phi(\lambda,x)&=\begin{pmatrix}1\\0\end{pmatrix}\E^{-\I\lambda x}-\int^{x}_{-\infty} L(x,s)\E^{-\I\lambda s}ds\\
\overline{\phi}(x)&=-\begin{pmatrix}0\\1\end{pmatrix}\E^{\I\lambda x}-\int^{x}_{-\infty} \overline{L}(x,s)\E^{\I\lambda s}ds,\\
\end{split}
\end{equation}
where  the integral kernels $K,\overline{K}, L, \overline{L}$  exist and are unique. This follows using the expression \eqref{NLSEq:001} into \eqref{NLSEq:2} 
\begin{equation*}
\begin{split}
&(\partial_x-\partial_s)K_1(x,s)-q(x)K_2(x,s)=0\\
&(\partial_x+\partial_s)K_2(x,s)-r(x)K_1(x,s)=0,
\end{split}
\end{equation*}
together with the boundary conditions $K_1(x,x)=-1/2q(x)$ and $\lim_{s\to\infty}K(x,s)=0$. Similar equations can be obtained for the other Kernels.

We are now in conditions to describe the Marchenko type equations which allow us to recover the potentials via the scattering data:
\begin{equation}\label{NLSEq:12}
\begin{split}
&\overline{K}(x,y)+\begin{pmatrix}0\\1\end{pmatrix}F(x+y)+\int_x^\infty K(x,s)F(s+y)ds=0,\quad (y>x)\\
&K(x,y)-\begin{pmatrix}1\\0\end{pmatrix}\overline{F}(x+y)-\int_x^\infty \overline{K}(x,s)\overline{F}(s+y)ds=0,\quad (y>x)\\
&\overline{L}(x,y)+\begin{pmatrix}1\\0\end{pmatrix}G(x+y)+\int_{-\infty}^x L(x,s)G(s+y)ds=0,\quad (x>y)\\
&L(x,y)+\begin{pmatrix}0\\1\end{pmatrix}\overline{G}(x+y)+\int_{-\infty}^x \overline{L}(x,s)\overline{G}(s+y)ds=0,\quad (x>y),
\end{split}
\end{equation}
where
\begin{equation*}
\begin{split}
F(z)&=\frac{1}{2\pi}\int_\Gamma \frac{b(\lambda)}{a(\lambda)}\E^{\I\lambda z}d\lambda,\\
\overline{F}(z)&=\frac{1}{2\pi}\int_{\overline{\Gamma}}\frac{\overline{b}(\lambda)}{\overline{a}(\lambda)}\E^{-\I\lambda z}d\lambda,\\
G(z)&=\frac{1}{2\pi}\int_\Gamma \frac{\overline{b}(\lambda)}{a(\lambda)}\E^{-\I\lambda z}d\lambda,\\
\overline{G}(z)&=\frac{1}{2\pi}\int_{\overline{\Gamma}}\frac{b(\lambda)}{\overline{a}(\lambda)}\E^{\I\lambda z}d\lambda.
\end{split}
\end{equation*}
Here $\Gamma$ denotes the contour in the complex $\lambda-$plane, starting from $\lambda=-\infty+\I0^+$, passing over all zeros of $a(\lambda)$, and ending at $\lambda=+\infty+\I0^+$. Similarly, $\overline{\Gamma}$ is the contour which starts from $\lambda=-\infty+\I0^-$, passes over all zeros of $\overline{a}(\lambda)$ and ends at $\lambda=+\infty+\I0^-$. Thus, we can write the potentials $q,r$ in terms of the integral Kernels in \eqref{NLSEq:001}
\begin{equation}\label{NLSEq:14}
\begin{split}
&K_1(x,x)=-\overline{L}_1(x,x)=-\frac{1}{2}q(x),\\
&K_2(x,x)=\overline{K}_1(x,x)=\frac{1}{2}\int_x^\infty q(y)r(y)dy,\\
&L_1(x,x)=-\overline{L}_2(x,x)=\frac{1}{2}\int_{-\infty}^x q(y)r(y)dy,\\
&L_2(x,x)=\overline{K}_2(x,x)=\frac{1}{2}r(x).
\end{split}
\end{equation}

Finally, since in general, it may happen that the equations \eqref{NLSEq:12} have more than one solution when the potentials $q,r$ evolve through time, we will assume that there always exists a solution for \eqref{NLSEq:12} and that this solution is unique. This behaviour can occur when $A_0(\lambda), q(x,0),r(x,0)$ are unrestricted. For instance, if we are working with the Korteweg-de Vries equation, then this cannot happen because the constraint on the initial data $\int_\mathbb{R}(1+|x|)|u|dx<\infty$ remains valid for any future time.

To finish this subsubsection, we want to remark that the Jost solutions can be written differently from \eqref{NLSEq:001}. These expressions (see below) together with \eqref{NLSEq:7} allow us to study if the scattering coefficients $a(\lambda),\overline{a}(\lambda),b(\lambda),\overline{b}(\lambda)$ can be extended outside the real line or if they are analytic functions in some region of the complex plane:
\begin{equation}\label{NLSEq:15}
\begin{split}
\phi_1(x)&=\E^{-\I\lambda x}+\E^{-\I\lambda x}\int_{-\infty}^x M(\lambda,x,y)\E^{\I\lambda y}\phi_1(y)dy\\ 
\phi_2(x)&=\E^{\I\lambda x}\int_{-\infty}^x\E^{-\I \lambda y}r(y)\phi_1(y)dy\\
\overline{\psi}_1(x)&=\E^{-\I\lambda x}+\E^{-\I\lambda x}\int_x^\infty\overline{M}(\lambda,x,y)\E^{\I\lambda y}\overline{\psi}(y)dy\\ 
\overline{\psi}_2(x)&=-\E^{\I\lambda x}\int_x^{\infty}r(y)\E^{-\I\lambda y}\overline{\psi}_1(y)dy\\
\psi_1(x)&=\E^{-\I\lambda x}\int_x^{\infty} \Big(\E^{\I\lambda y}\psi_1(y)\overline{M}(\lambda,x,y)-q(y)\E^{2\I\lambda y}\Big)dy\\
\psi_2(x)&=-\E^{\I\lambda x}\Big(-1+\int_x^\infty r(y)\psi_1(y)\E^{-\I\lambda y}dy\Big)\\
\overline{\phi}_1(x)&=\E^{-\I\lambda x}\int_{-\infty}^x\E^{\I\lambda y}q(y)\overline{\phi}_2(y)dy\\
\overline{\phi}_2(x)&=\E^{\I\lambda x}\Big(-1+\int_{-\infty}^xM^*(\lambda,x,y)\overline{\phi}_2(x)\E^{-\I\lambda y}dy\Big),
\end{split}
\end{equation}
where 
\begin{equation*}
\begin{split}
&M(\lambda,x,y)=r(y)\int_y^x\E^{2\I\lambda(z-y)}q(z)dz,\quad \overline{M}(\lambda,x,y)=r(y)\int_{x}^y\E^{2\I\lambda(z-y)}q(z)dz \\
&\text{and }\quad M^*(\lambda,x,y)=q(y)\int_y^x\E^{2\I\lambda(y-z)}r(z)dz.
\end{split}
\end{equation*}
To summarize, these are the main facts of the {\em NLS case}:
\begin{itemize}
\item[i)]$\int_{\mathbb{R}}(1+|x|)|q(x,t)|dx, \int_{\mathbb{R}}(1+|x|)|r(x,t)|dx<\infty$.
\item[ii)]$A_0(\lambda)=A_-(\lambda)=A_+(\lambda)$, $B_-=B+=C_-=C_+=0.$
\item[iii)]The system of equations \eqref{NLSEq:12} has a unique solution for all time $t$.
\item[iv)]The zeros (bound states) of $a(\lambda)$, $\{\lambda_k\}_{k\geq 1}^N$, and $\overline{a}(\lambda)$, $\{\overline{\lambda}_k\}_{k\geq1}^{\overline{N}}$, are finite, independent of time and $\im(\lambda_k)>0$, $\im(\overline{\lambda}_k)<0$.  In general, $N\neq\overline{N}$ and $\lambda_k$ and $\overline{\lambda}_k$ are not related.
\end{itemize}
As we pointed out in the introduction if we set $r=-q^*$ and $A_0(z)=-2\I z^2$, then we obtain the nonlinear Schr\"odinger equation and if  we set $r=q$ and $A_0(z)=-4\I z^3$, then we have the modified Korteweg-de Vries equation. Moreover, if we consider the linear part of these two equations, we observe that $w(k)=-\I k^2$ for the nonlinear Schr\"odinger equation and $w(k)=-\I k^3$ for the modified Korteweg-de Vires equation, where $w(k)$ denotes the dispersion relation. We remark how close these numbers $w(k)$ and $A_0(z)$ are.

\subsubsection{KdV case}\label{SecKdV}
We will give a brief description of the special case when $r=-1$. We will follow the Appendix 3 in \cite{AKNS} and \cite{DT} for the study of the Schr\"odinger equation.

Since $r=-1$,  \eqref{NLSEq:1} is equivalent to the Schr\"odinger equation
\begin{equation}\label{KdVEq:0}
v_{2xx}+(\lambda^2+q)v_2=0.
\end{equation}
Here the appropriate eigenfunctions are $\phi,\overline{\phi}$, which have the following asymptotic behaviour
\begin{equation*}
\begin{split}
&\phi\to\begin{pmatrix}2\I\lambda\\1\end{pmatrix}\E^{-\I\lambda x},\\
&\overline{\phi}\to\begin{pmatrix}0\\1\end{pmatrix}\E^{\I\lambda x},\quad\text{as }x\to-\infty.
\end{split}
\end{equation*}
Thus, the scattering data appears if we let $x\to\infty$:
\begin{equation*}
\begin{split}
&\phi\to\begin{pmatrix}2\I\lambda a\E^{-\I\lambda x}\\a\E^{-\I\lambda x}+b\E^{\I\lambda x}\end{pmatrix},\\
&\overline{\phi}\to\begin{pmatrix}2\I\lambda\overline{b}\E^{-\I\lambda x}\\ \overline{a}\E^{\I\lambda x}+\overline{b}\E^{-\I\lambda x}\end{pmatrix},\quad \text{as }x\to\infty.
\end{split}
\end{equation*}
As we did in \eqref{NLSEq:9}, we observe that the functions $\phi\E^{A_-(\lambda)t}$ and $\overline{\phi}\E^{-A_-(\lambda)t}$ solve both equations \eqref{NLSEq:2} and \eqref{NLSEq:3}. This implies that once again we obtain similar expressions to  \eqref{NLSEq:9}, that is
\begin{equation*}
\begin{split}
\phi_t&=\begin{pmatrix}A-A_-& B\\ C& -A-A_-\end{pmatrix}\phi\\
\overline{\phi}_t&=\begin{pmatrix}A+A_-& B\\ C& -A+A_-\end{pmatrix}\overline{\phi},
\end{split}
\end{equation*}
and for the evolution equations for the scattering data, we will focus only in the special case  $A_+=A_-=\I\lambda C_+=\I\lambda C_-$ and $B_+=B_-=0$. Thus, as in \eqref{NLSEq:10}
\begin{equation*}
\begin{split}
&a_t=0\\
&b_t=-2A_+b\\
&\overline{a}_t=0\\
&\overline{b}_t=2A_+\overline{b}.
\end{split}
\end{equation*}
Notice that the usual transmission and reflection coefficients for the Schr\"odinger equation are given by
\begin{equation*}
\begin{split}
&T(\lambda)=\frac{1}{a(\lambda)}\\
&R(\lambda)=\frac{b(\lambda)}{a(\lambda)}.
\end{split}
\end{equation*}
Finally, focusing on the Schr\"odinger equation \eqref{KdVEq:0}, and not on the evolution, we can just follow the analysis which is done in \cite{DT} by P. Deift and E. Trubowitz. We will be interested in the Marchenko type equations and the integral formulas for the Jost solutions and the transmission and reflection coefficients. 

We will assume that $\int_{\mathbb{R}}(1+|x|)|q|dx<\infty$. Thus,  if we set $m_1=\E^{-\I\lambda x}f_1$ and $m_2=\E^{\I\lambda x}f_2$, where $f_1,f_2$ are  the Jost solutions for \eqref{KdVEq:0},  whose asymptotic behaviours are $f_1\to\E^{\I\lambda x}$, as $x\to\infty$ and $f_2\to\E^{-\I\lambda_x}$, as $x\to\-\infty$, then, see Lemma 1 in \cite{DT},

\begin{equation}\label{KdVEq:4}
\begin{split}
&m_1(x,\lambda)=1-\frac{1}{2\I\lambda}\int_{x}^\infty q(t)m_1(t,k)(\E^{2\I\lambda(t-x)}-1)dt\\
&m_1(x,\lambda)=1-\frac{1}{2\I\lambda}\int_{-\infty}^x q(t)m_2(t,k)(\E^{2\I\lambda(x-t)}-1)dt,
\end{split}
\end{equation}
where $\im(\lambda)\geq 0$.

In addition, by Lemma 3 in \cite{DT}, we have that there exist $B_1(x,y), B_2(x,y)$ such that 
\begin{equation*}
\begin{split}
&m_1(x,\lambda)=1+\int_0^\infty B_1(x,y)\E^{2\I\lambda y}dy\\
&m_2(x,\lambda)=1+\int_{-\infty}^0 B_2(x,y)\E^{-2\I\lambda y}dy,
\end{split}
\end{equation*}
where $B_1,B_2$ satisfy:
\begin{equation*}
\begin{split}
&B_1(x,y)=\int_{x+y}^\infty q(t)dt+\int_0^y\int_{x+y-z}^\infty q(t)B_1(t,z)dzdt\\
&B_2(x,y)=\int_{-\infty}^{x+y}q(t)dt+\int_y^0\int_{-\infty}^{x+y-z}q(t)B_2(t,z)dtdz
\end{split}
\end{equation*}
and hence
\begin{equation}
-\frac{\partial B_1(x,0+)}{\partial x}=-\frac{\partial B_1(x,0+)}{\partial y}=\frac{\partial B_2(x,0-)}{\partial x}=\frac{\partial B_2(x,0-)}{\partial y}=q(x)
\end{equation}

For $\lambda\in\mathbb{R}\setminus\{0\}$, $f_1(x,\lambda)$ and $f_1(x,-\lambda)$ are two independent solutions, as it follows from the Wronskian:
\begin{equation*}
[f_1(x,\lambda),f_1(x,-\lambda)]=f^\prime_1(x,\lambda)f_1(x,-\lambda)-f_1(x,\lambda)f_1^\prime(x,-\lambda)=2\I\lambda\neq 0.
\end{equation*}
Similarly, $[f_2(x,\lambda),f_2(x,-\lambda)]=-2\I\lambda\neq 0$. Thus, the transmission and reflection coefficients appear naturally and 

\begin{equation*}
\begin{split}
&f_2(x,\lambda)=\frac{R_1(\lambda)}{T(\lambda)}f_1(x,\lambda)+\frac{1}{T(\lambda)}f_1(x,-\lambda)\\
&f_1(x,\lambda)=\frac{R_2(\lambda)}{T(\lambda)}f_2(x,\lambda)+\frac{1}{T(\lambda)}f_2(x,-\lambda),
\end{split}
\end{equation*}
and hence

\begin{equation}\label{KdVEq:5}
\begin{split}
&T(\lambda)m_2(x,\lambda)=R_1(\lambda)\E^{2\I\lambda x}m_1(x,\lambda)+m_1(x,-\lambda),\\
&T(\lambda)m_1(x,\lambda)=R_2(\lambda)\E^{-2\I\lambda x}m_2(x,\lambda)+m_2(x,-\lambda).
\end{split}
\end{equation}
Combining these two last expressions \eqref{KdVEq:4} and \eqref{KdVEq:5}, we obtain the integral form for the scattering coefficients:
\begin{equation}\label{KdVEq:6}
\begin{split}
&\frac{R_2(\lambda)}{T(\lambda)}=\frac{-1}{2\I\lambda}\int_{\mathbb{R}}\E^{2\I\lambda t}q(t)m_1(t,\lambda)dt,\\
&\frac{1}{T(\lambda)}=1+\frac{1}{2\I\lambda}\int_{\mathbb{R}}q(t)m_1(t,\lambda)dt\\
&\frac{1}{T(\lambda)}=1+\frac{1}{2\I\lambda}\int_{\mathbb{R}}q(t)m_2(t,\lambda)dt\\
&\frac{R_1(\lambda)}{T(\lambda)}=\frac{-1}{2\I\lambda}\int_{\mathbb{R}}\E^{-2\I\lambda t}q(t)m_2(t,\lambda)dt.
\end{split}
\end{equation}
Finally, all the eigenvalues of \eqref{KdVEq:0} comes from the zeros of $1/T(\lambda)$:  $\{\I\beta_n\}_{n=1}^N,$ being all of them simple zeros and $\beta_n>0$. We are now in conditions to establish the Marchenko type equations:

\begin{equation}\label{Revision01}
\begin{split}
&F_1(x+y)+B_1(x,y)+\int_0^\infty F_1(x+y+t)B_1(x,t)dt=0\\
&F_2(x+y)+B_2(x,y)+\int_{-\infty}^0F_2(x+y+t)B_2(x,t)dt=0,
\end{split}
\end{equation}
where $F_1(x)=\pi^{-1}\int_\mathbb{R}R_1\E^{2\I\lambda x}d\lambda+2\sum_{n=1}^Nc_n\E^{-2\beta_n x}$,$F_2(x)=\pi^{-1}\int_\mathbb{R}R_2\E^{-2\I\lambda x}d\lambda+2\sum_{n=1}^Nc_n\E^{2\beta_n x}$, and  $c_n=\Big(\int f_1^2(x,\I\beta_n)dx\Big)^{-1}$ are the norming constants.

To summarize, these are the main facts of the {\em KDV case}:
\begin{itemize}
\item[i)] $r=-1$.
\item[ii)]$\int_{\mathbb{R}}(1+|x|)|q(x,t)|dx<\infty$.
\item[iii)]$A_0(\lambda)=A_-(\lambda)=A_+(\lambda)=\I\lambda C_+=\I\lambda C_-$, $B_-=B+=0$.
\item[iv)]The system of equations \eqref{Revision01} has a unique solution for all time $t$.
\item[v)]The zeros (bound states) of $a(\lambda)$, $\{\lambda_k\}_{k\geq 1}^N$  are finite, independent of time and $\im(\lambda_k)>0$.
\end{itemize}
In addition, if we set $A_0(z)=-4\I z^3$ or $A_0(z)=-\I z$, then we obtain the Korteweg-de Vries equation and the first term of the Korteweg-de Vries hierarchy ($u_t=u_x$), respectively. As we observed in the previous subsection, if we focus on the linear part of these equations, then $w(k)=-\I k^3$ (Korteweg-de Vries equation) and $w(k)=\I k$ ($u_t=u_x$) and the similarities of $A_0(z)$ and $w(k)$ are obvious.

\section{Main result}

As it is shown in \cite{DT}, we observe that from \eqref{NLSEq:6} and \eqref{NLSEq:15} we obtain an integral representation for the scattering coefficients. We are interested only in $b(\lambda),\overline{b}(\lambda)$:
\begin{equation*}
\begin{split}
\E^{-\I\lambda x}\phi_2(x)&=\int_{-\infty}^x\E^{-\I\lambda y}r(y)\phi_1(y)dy\to\int_{-\infty}^\infty\E^{-\I\lambda y}r(y)\phi_1(y)dy,\quad \text{and }\\
\E^{-\I\lambda x}\phi_2(x)&\to b(\lambda)\quad \text{as }x\to\infty,\\
\end{split}
\end{equation*}
thus
\begin{equation}\label{MREq:2}
b(\lambda)=\int_{-\infty}^\infty\E^{-\I\lambda y}r(y)\phi_1(y)dy.
\end{equation}
Similarly we obtain
\begin{equation}\label{MREq:3}
\overline{b}(\lambda)=\int_{-\infty}^\infty \E^{\I\lambda y}q(y)\overline{\phi}_2(y)dy.
\end{equation}
Notice that in general $\lambda\in\mathbb{R}$. Thus, to extend these formulas beyond the real line, we need some decay of the potentials $q(x),r(x)$ as $|x|\to\infty$.

\begin{lemma}\label{Lemma:1}
Let $(q,r)$ be a solution for the coupled equations in \eqref{NLSEq:1} and assume that we are in the NLS case. If for some $t_0$ :
\begin{itemize}
\item[i)]   $\int_{\mathbb{R}}(1+|x|)|q(x,t_0)|dx<\infty$ and $\int_{\mathbb{R}}(1+|x|)|r(x,t_0)|dx<\infty$,
\item[ii)] $|q(x,t_0)|\leq C_1\E^{-C_2 x^{1+\beta}} $, $|r(x,t_0)|\leq C_3\E^{-C_4 x^{1+\delta}}$, where $C_i,\delta,\beta>0$  are constants $i=1,2,3,4$ and $x\geq 0$, and
\item[iii)] the system is reflectionless, that is $b(\lambda,t_0)=\overline{b}(\lambda,t_0)=0$,
\end{itemize}
then $q(x,t)=r(x,t)=0$ for all $x,t\in\mathbb{R}$.
\end{lemma}

\begin{proof}

We will follow a similar strategy as in \cite{Mar}, Chapter 4, section 2, to show how the solution $(q,r)$ looks like. Thus, we will use the first two equations of \eqref{NLSEq:12}, since by \eqref{NLSEq:14}, we know $K_1(x,x)=-1/2q(x)$ and $\overline{K}_2(x,x)=1/2r(x)$.

By hypothesis, the potentials are reflectionless and hence we know that $F(z;t)=-\I\sum_{k=1}^N m_k\E^{\I\lambda z-2A_0(\lambda_k)t}$ and $\overline{F}(z;t)=\I\sum_{k=1}^{\overline{N}}\overline{m}_k\E^{-\I\overline{\lambda}_kz+2A_0(\overline{\lambda}_k)t}$, where $m_k,\overline{m}_k$ are constants which are related with the residues of $1/a(\lambda)$ and $1/\overline{a}(\lambda)$ respectively. Setting 
\begin{equation*}
\begin{split}
&P_k(x)=\begin{pmatrix}P_k^{(1)}(x)\\P_k^{(2)(x)}\end{pmatrix}=\I\overline{m}_k\Big\{\begin{pmatrix}1\\0\end{pmatrix}\E^{-\I\overline{\lambda}_k x}+\int_x^\infty\overline{K}(x,s;t)\E^{-\I\overline{\lambda}_ks}ds\Big\},\quad 1\leq k\leq \overline{N}\\
&Q_k(x)=\begin{pmatrix}Q_k^{(1)}(x)\\Q_k^{(2)}(x)\end{pmatrix}=\I m_k\Big\{\begin{pmatrix}0\\1\end{pmatrix}\E^{\I\lambda_k x}+\int_x^\infty K(x,s;t)\E^{\I\lambda_ks}ds\Big\},\quad 1\leq k\leq N,
\end{split}
\end{equation*}
and using \eqref{NLSEq:12} we arrive to
\begin{equation}\label{EqLemma:2}
\begin{split}
&K(x,y;t)=\sum_{k=1}^{\overline{N}}\E^{-\I\overline{\lambda}_ky+2A_0(\overline{\lambda}_k)t}P_k(x)\\
&\overline{K}(x,y;t)=\sum_{k=1}^{N}\E^{\I\lambda_ky-2A_0(\lambda_k)t}Q_k(x).
\end{split}
\end{equation}
Substituting this expression into \eqref{NLSEq:12}, we obtain the identities
\begin{equation*}
\begin{split}
&\sum_{k=1}^{\overline{N}}\E^{-\I\overline{\lambda}_ky+2A_0(\overline{\lambda}_k)t}\Big\{P_k(x)-\I\overline{m}_k\begin{pmatrix}1\\0\end{pmatrix}\E^{-\I\overline{\lambda}_kx}\\
&-\I\overline{m}_k\sum_{l=1}^NQ_l(x)\E^{-2A_0(\lambda_l)t}\frac{\E^{(\I\lambda_l-\I\overline{\lambda}_k)x}}{\I\overline{\lambda}_k-\I\lambda_l}\Big\}=0,\quad \text{and}\\
&\sum_{k=1}^N\E^{\I\lambda_k-2A_0(\lambda_k)t}\Big\{Q_k(x)-\I m_k\begin{pmatrix}0\\1\end{pmatrix}\E^{\I\lambda_k x}\\
&-\I m_k\sum_{l=1}^{\overline{N}}P_l(x)\E^{2A_0(\overline{\lambda}_l)t}\frac{\E^{(\I\lambda_k-\I\overline{\lambda}_l)x}}{\I\overline{\lambda}_l-\I\lambda_k}\Big\}=0. 
\end{split}
\end{equation*}
Thus, we have the following  equations:
\begin{equation}\label{EqLemma:4}
P_k(x)-\I\overline{m}_k\begin{pmatrix}1\\0\end{pmatrix}\E^{-\I\overline{\lambda}_kx}-\I\overline{m}_k\sum_{l=1}^NQ_l(x)\E^{-2A_0(\lambda_l)t}\frac{\E^{(\I\lambda_l-\I\overline{\lambda}_k)x}}{\I\overline{\lambda}_k-\I\lambda_l}=0
\end{equation}
\begin{equation}\label{EqLemma:41}
Q_k(x)-\I m_k\begin{pmatrix}0\\1\end{pmatrix}\E^{\I\lambda_k x}-\I m_k\sum_{l=1}^{\overline{N}}P_l(x)\E^{2A_0(\overline{\lambda}_l)t}\frac{\E^{(\I\lambda_k-\I\overline{\lambda}_l)x}}{\I\overline{\lambda}_l-\I\lambda_k}=0. 
\end{equation}
Substituting the value of $Q_k(x)$ from \eqref{EqLemma:41} into   \eqref{EqLemma:4}, we observe 
\begin{equation*}
\begin{split}
&P_k-\sum_{j=1}^{\overline{N}}P_j(x)\sum_{l=1}^N\frac{m_l\overline{m}_k\E^{2(A_0(\overline{\lambda}_j)-A_0(\lambda_l))t}}{(\lambda_l-\overline{\lambda}_k)(\lambda_l-\overline{\lambda}_j)}\E^{\I(2\lambda_l-\overline{\lambda}_k-\overline{\lambda}_j)x}=\\
&\I\overline{m}_k\begin{pmatrix}1\\0\end{pmatrix}\E^{-\I\overline{\lambda}_kx}-\I\sum_{l=1}^N\frac{m_l\overline{m}_k\E^{-2A_0(\lambda_l)t}}{\lambda_l-\overline{\lambda}_k}\E^{\I(2\lambda_l-\overline{\lambda}_k)x}\begin{pmatrix}0\\1\end{pmatrix},\quad 1\leq k\leq\overline{N},
\end{split}
\end{equation*}
and analogously  for $Q_k(x)$
\begin{equation*}
\begin{split}
&Q_k(x)-\sum_{j=1}^N Q_j(x)\sum_{l=1}^{\overline{N}}\frac{m_k\overline{m}_l\E^{2(A_0(\overline{\lambda}_l)-A_0(\lambda_j))t}}{(\lambda_j-\overline{\lambda}_l)(\lambda_k-\overline{\lambda}_l)}\E^{\I(\lambda_k+\lambda_j-2\overline{\lambda}_l)x}=\\
&\I m_k\begin{pmatrix}0\\1\end{pmatrix}\E^{\I\lambda_kx}-\I\sum_{l=1}^{\overline{N}}\frac{m_k\overline{m}_l\E^{2A_0(\overline{\lambda}_l)t}}{\lambda_k-\overline{\lambda}_l}\E^{\I(\lambda_k-2\overline{\lambda}_l)x}\begin{pmatrix}1\\0\end{pmatrix},\quad 1\leq k\leq N.
\end{split}
\end{equation*}
Thus, using the Cramer's rule
\begin{equation*}
\begin{split}
&P_k^{(1)}(x)=\Delta_k^P[\Delta_P]^{-1}\\
&Q_k^{(2)}(x)=\Delta_k^Q[\Delta_Q]^{-1},
\end{split}
\end{equation*}
where 
\begin{equation}\label{EqLemma:7}
\begin{split}
&\Delta_P=\det\Big(\delta_{kj}+\sum_{l=1}^N\frac{m_l\overline{m}_k\E^{2(A_0(\overline{\lambda}_j)-A_0(\lambda_l))t}}{(\lambda_l-\overline{\lambda}_k)(\lambda_l-\overline{\lambda}_j)}\E^{\I(2\lambda_l-\overline{\lambda}_k-\overline{\lambda}_j)x}\Big)_{1\leq k,j\leq \overline{N}}\\
&\Delta_Q=\det\Big(\delta_{kj}+\sum_{l=1}^{\overline{N}}\frac{m_k\overline{m}_l\E^{2(A_0(\overline{\lambda}_l)-A_0(\lambda_j))t}}{(\lambda_j-\overline{\lambda}_l)(\lambda_k-\overline{\lambda}_l)}\E^{\I(\lambda_k+\lambda_j-2\overline{\lambda}_l)x}\Big)_{1\leq k,j\leq N}.
\end{split}
\end{equation}
$\Delta_k^P$ is the determinant of the matrix, obtained from $\Delta_P$ upon replacing the $j-$column by $\I \overline{m}_k\E^{-\I\overline{\lambda}_kx}$, respectively, $\Delta_k^Q$ is obained from $\Delta_Q$, but in this case we substitute the $j-$column by $\I m_k\E^{\I\lambda_kx}$.

Now, using \eqref{EqLemma:2} and without lost of generality, we will assume that $t_0=0$
\begin{equation}\label{EqLemma:8}
\begin{split}
&k_1(x,x)=\frac{1}{\Delta_P}\sum_{k=1}^{\overline{N}}\Delta_k^P\E^{-\I\overline{\lambda}_kx}\\
&\overline{k}_2(x,x)=\frac{1}{\Delta_Q}\sum_{k=1}^N\Delta_k^Q\E^{\I\lambda_kx}.
\end{split}
\end{equation}
Notice that $\Delta_P,\Delta_Q\to1$ as $x\to\infty$ and \eqref{EqLemma:8} can be written in terms of the sum of exponentials, that is
\begin{equation}
\begin{split}
&k_1(x,x)=\frac{1}{\Delta_P}\sum_{k=1}^{\overline{N}^\prime}\E^{\I\alpha_kx}\\
&\overline{k}_2(x,x)=\frac{1}{\Delta_Q}\sum_{k=1}^{N^\prime}\E^{\I\beta_k x},
\end{split}
\end{equation}
where in general $\overline{N}^\prime\neq\overline{N}, N^\prime\neq N$. Moreover, we observe that $\im(\alpha_k),\im(\beta_k)>0$, this is due to $\im(\lambda_k)>0$ and $\im(\overline{\lambda}_k)<0$. Thus,  there exists an increasing sequence $\{x_n\}_{n\geq 1}$, such that $x_n\to\infty$, as $n\to\infty$ and for $n$ big enough, we have
\begin{equation}\label{EqLemma:9}
C\frac{\E^{-x_n\im(\alpha_{k_0})}}{1+D}\leq |k_1(x_n,x_n)|=\frac{1}{2}|q(x_n,0)|\leq C_1\E^{-C_2x_n^{1+\delta}},
\end{equation}
here $C,D>0$ are constants independent of $x_n$, but \eqref{EqLemma:9} leads us a contradiction unless $q(x,0)$ has no bound states. Similar calculations can be done for $k_2(x_n,x_n)=\frac{1}{2}r(x)$ and hence, it must be $r(x,0)=q(x,0)=0$. Since the bound states (eigenvalues) are independent of time, this means that there are no bound states for any time $t$. Using the evolution of the scattering data \eqref{NLSEq:11}, we observe that $b(\lambda,t)=\overline{b}(\lambda,t)=0$ for any time $t$. This means that $F(z;t)=\overline{F}(z;t)=0$ and hence, by \eqref{NLSEq:12}, $K(x,y;t)=\overline{K}(x,y;t)=0$ for any time $t$. In particular,  we have proved that $q(x,t)=r(x,t)=0$ for all $x,t\in\mathbb{R}$.

\end{proof}

\begin{theorem}\label{theorem:1}
Let $(q,r)$ be a solution for the coupled equations in \eqref{NLSEq:1}.  Suppose  that we are in the NLS case and at two different times $t_0<t_1$ we also have the following properties:
\begin{itemize}
\item[i)]   $\int_{\mathbb{R}}(1+|x|)|q(x,t_i)|dx<\infty$ and $\int_{\mathbb{R}}(1+|x|)|r(x,t_i)|dx<\infty$, for $i=0,1$, and
\item[ii)] $|r(x,t_j)|\leq C_1\E^{-C_2 x^{1+\delta}}$ and  $|q(x,t_j)|\leq C_3\E^{-C_4 x^{1+\beta}} $, where $C_i,\delta,\beta>0$,  $i=1,2,3,4$, $j=0,1$, are constants and  $x\geq 0$. The constants  $\delta,\beta$  also fulfill 
\begin{equation}\label{ThEq:1}
\limsup_{|\lambda|\to\infty}\frac{|\re(A_0(\lambda))|}{|\lambda|^\rho}=\infty,\quad\quad \rho>1+1/\delta, 1+1/\beta,
\end{equation}
\end{itemize}
 then $q(x,t)=r(x,t)=0$ for all $x,t\in\mathbb{R}$.
\end{theorem}
\begin{proof}
We will first prove that the scattering coefficients $b(\lambda,t_i),\overline{b}(\lambda,t_i)$, $i=0,1$, can be extended to the upper and lower half plane, respectively, and hence by \eqref{NLSEq:11} for all $t\in\mathbb{R}$. Then we will assume that these coefficients are different from zero and we will obtain a contradiction, which will  show that they must vanish and we will be able to apply Lemma \ref{Lemma:1}, proving the theorem.

Without lost of generality we can assume that $t_0=0$ and $t_1=1$. We begin with $b(\lambda,t_0)$, similar calculations can be applied to $b(\lambda,t_1)$. To make notation easier we will omit the parameter $t$. We know  by \eqref{MREq:2}  
\begin{equation}\label{MREq:4}
|b(\lambda)|\leq\int_{-\infty}^\infty|\E^{-\I\lambda y}r(y)\phi_1(y)|dy=\int_{-\infty}^\infty\E^{\im(\lambda)y}|r(y)\phi_1(y)|dy.
\end{equation}
Thus,  we  need to estimate first $\phi_1(\cdot)$. Using the same idea as in the proof of Lemma 1 in \cite{DT} and the integral representation \eqref{NLSEq:15}, we set
\begin{equation*}
\begin{cases}g_0(x)&=1\\ g_n(x)&=\int_{-\infty}^x M(\lambda,x,y)g_{n-1}(y)dy,\quad n\geq 1,\end{cases}
\end{equation*}
and 
\begin{equation}\label{MREq:5}
\E^{\I\lambda x}\phi_1(x)=\sum_{n\geq 0}g_n(x).
\end{equation}
To estimate $g_n(\cdot)$, we need first to estimate the kernel $M(\lambda,x,y)$:
\begin{equation*}
\begin{split}
|M(\lambda,x,y)|&\leq |r(y)|\int_y^x|\E^{2\I\lambda (z-y)}q(z)|dz=|r(y)|\int_y^x\E^{-2\im(\lambda )(z-y)}|q(z)|dz\\
&\leq|r(y)|\int_y^x|q(z)|dz\leq |r(y)|\int_{\mathbb{R}}|q(z)|dz=|r(y)|Q_0.
\end{split}
\end{equation*}
We have assumed that $\im(\lambda)\geq 0$. Thus, $-2\im(\lambda )(z-y)\leq 0$, since $z\geq y$, and we denote $Q_0=\int_{\mathbb{R}}|q(z)|dz<\infty$.

Using a simple induction, we arrive to
\begin{equation}\label{MREq:6}
|g_n(x)|\leq\frac{1}{n!}\Big(Q_0\int_{-\infty}^x|r(y)|dy\Big)^n.
\end{equation}
Combining \eqref{MREq:5} and \eqref{MREq:6}
\begin{equation*}
\begin{split}
|\phi_1(x)|&\leq |\E^{-\I\lambda x}|\sum_{n\geq 0}|g_n(x)|\leq \E^{\im(\lambda)x}\E^{Q_0\int_{-\infty}^x|r(y)|dy}\\
&\leq\E^{\im(\lambda)x+Q_0R_0},
\end{split}
\end{equation*}
where $R_0=\int_{\mathbb{R}}|r(x)|dx<\infty$. This means that \eqref{MREq:4} turns into
\begin{equation*}
\begin{split}
|b(\lambda)|&\leq D\int_{\mathbb{R}}\E^{\im(\lambda)y}|r(y)|\E^{\im(\lambda)y}dy\\
&=D\Big(\int_{-\infty}^0|r(y)|\E^{2\im(\lambda)y}dy+\int^{\infty}_0|r(y)|\E^{2\im(\lambda)y}dy\Big)\\
&\leq D\Big(\int_{-\infty}^0|r(y)|dy+\int^{\infty}_0|r(y)|\E^{2\im(\lambda)y}dy\Big)=D(I_1+I_2),
\end{split}
\end{equation*}
where $D=\exp(Q_0R_0)$ is a constant.  The first integral $I_1$ is trivially bounded by $R_0$. To see that $I_2$ is bounded as well, we need to use the exponential decay of $r(x)$:
\begin{equation*}
I_2\leq\int_0^\infty C_1\E^{-C_2y^{1+\delta}}\E^{2\im(\lambda)}dy=C_1\int_0^\infty\E^{y(2\im(\lambda)-C_2y^\delta)}dy.
\end{equation*}
Now, if $s\geq0$, then $f(s)=2\im(\lambda)-C_2s^\delta=0$ if and only if $s=[2\im(\lambda)/C_2]^{1/\delta}$. Thus,
\begin{equation*}
\begin{split}
I_2&\leq C_1\Big(\int_0^{[2\im(\lambda)/C_2]^{1/\delta}}\E^{y(2\im(\lambda)-C_2y^\delta)}dy+\int_{[2\im(\lambda)/C_2]^{1/\delta}}^\infty\E^{y(2\im(\lambda)-C_2y^\delta)}dy\Big)\\
&=C_1(J_1+J_2).
\end{split}
\end{equation*}
Since $\lambda$ is a fixed complex number, this means that $J_1<\infty$. Moreover $\lim_{s\to\infty}f(s)=-\infty$, thus there is an $s_0$, such that $f(s)<-N$ for all $s>s_0$ and $N\geq 1$ a constant. This proves the boundness of $J_2$ and hence that $b(\lambda,0)$ can be extended to the uper half plane  as a meromorphic function and using \eqref{NLSEq:11} for all t.

In order to prove that actually $b(\lambda)=0$, we need to bound  $J_1,J_2$ in terms of $\im(\lambda)$. A simple calculation shows us that the function $sf(s)$, $s\geq 0$, has a maximum at $s=[2\im(\lambda)/C_2(1+\delta)]^{1+\delta}$. Thus,
\begin{equation*}
\begin{split}
J_1&=\int_0^{[2\im(\lambda)/C_2]^{1/\delta}}\E^{sf(s)}ds\\
&\leq [2\im(\lambda)/C_2]^{1/\delta}\exp\Big\{\Big(\frac{2\im(\lambda)}{C_2(1+\delta)}\Big)^{1+\frac{1}{\delta}}C_2\delta\Big\}.
\end{split}
\end{equation*}

For $J_2$, we remark  that the function $f(s)$ is decreasing for $s\geq[2\im(\lambda)/C_2]^{1/\delta}$. Thus,  let $s_0$ be the number such that $2\im(\lambda)-C_2s_0^\lambda=-1$, then 
\begin{equation}\label{MREq:9}
J_2=\int_{[2\im(\lambda)/C_2]^{1/\delta}}^\infty\E^{sf(s)}ds=\int_{[2\im(\lambda)/C_2]^{1/\delta}}^{s_0}\E^{sf(s)}ds+\int_{s_0}^\infty \E^{sf(s)}ds.
\end{equation}
The first summand in \eqref{MREq:9} can be bounded using that $sf(s)$ has a maximum at $s=[2\im(\lambda)/C_2]^{1/\delta}$, for $s\geq [2\im(\lambda)/C_2]^{1/\delta}$
\begin{equation*}
\begin{split}
\int_{[2\im(\lambda)/C_2]^{1/\delta}}^{s_0}\E^{sf(s)}&\leq\Big([(2\im(\lambda)+1)/C_2]^{1/\delta}-[2\im(\lambda)/C_2]^{1/\delta}\Big)\\
&\leq [(2\im(\lambda)+1)/C_2]^{1/\delta}.
\end{split}
\end{equation*}
Using once more that the function $f(s)$ is decreasing for $s\geq [2\im(\lambda)/C_2]^{1/\delta}$, it follows that the second summand in \eqref{MREq:9} is bounded by
\begin{equation*}
\int_{s_0}^\infty\E^{sf(s)}ds\leq\int_{s_0}^\infty\E^{-s}ds\leq\int_0^\infty\E^{-s}= 1.
\end{equation*}
Hence, we have obtained
\begin{equation}\label{MREq:10}
J_2\leq [(2\im(\lambda)+1)/C_2]^{1/\delta}+1.
\end{equation}
Now, we are in conditions to study the indicator function of $b(\lambda)$, when $\im(\lambda)\geq 0$. Let $\rho$ be a constant as in \eqref{ThEq:1}. We suppose that $b(\lambda)\neq 0$ and we will obtain a contradiction.

Using \eqref{Eq:3}, we observe
\begin{equation*}
\begin{split}
\limsup_{|\lambda|\to\infty}\frac{\log|b(\lambda)|}{|\lambda|^\rho}&\leq \limsup_{|\lambda|\to\infty}\frac{\log|D(I_1+I_2|)}{|\lambda|^\rho}\\
&\leq\max\Big\{\limsup_{|\lambda|\to\infty}\frac{\log|I_1|}{|\lambda|^\rho},\limsup_{|\lambda|\to\infty}\frac{\log|I_2|}{|\lambda|^\rho}\Big\}\\
&\leq\max\Big\{0,\limsup_{|\lambda|\to\infty}\frac{\log|I_2|}{|\lambda|^\rho}\Big\}
\end{split}
\end{equation*}
and
\begin{equation*}
\limsup_{|\lambda|\to\infty}\frac{\log|I_2|}{|\lambda|^\rho}\leq\max\Big\{\limsup_{|\lambda|\to\infty}\frac{\log|J_1|}{|\lambda|^\rho},\limsup_{|\lambda|\to\infty}\frac{\log|J_2|}{|\lambda|^\rho}\Big\},
\end{equation*}
where
\begin{equation*}
\begin{split}
&\limsup_{|\lambda|\to\infty}\frac{\log|J_1|}{|\lambda|^\rho}\leq\limsup_{|\lambda|\to\infty}\frac{\log|(2\im(\lambda)/C_2)^{1/\delta}|}{|\lambda|^\rho}\\
&+\limsup_{|\lambda|\to\infty}\frac{\log\Big|\exp\{C_2\delta(2\im(\lambda)/C_2(1+\delta))^{1+1/\delta}\}\Big|}{|\lambda|^\rho}\\
&\leq C_2\delta\limsup_{|\lambda|\to\infty}\frac{\Big\{2\im(\lambda)/C_2(1+\delta)\Big\}^{1+1/\delta}}{|\lambda|^\rho}\\
&\leq C_2\delta\Big(\frac{2}{C_2(1+\delta)}\Big)\limsup_{|\lambda|\to\infty}\frac{|\lambda|^{1+1/\delta}}{|\lambda|^\rho}=0.
\end{split}
\end{equation*}
The last equality follows from \eqref{ThEq:1}, where we see that $\rho>1+1/\delta$. It remains to show how fast  $J_2$ decays as $|\lambda|\to\infty$. In order to do this, we will use \eqref{MREq:10}:
\begin{equation*}
\begin{split}
\limsup_{|\lambda|\to\infty}\frac{\log|J_2|}{|\lambda|^\rho}\leq\limsup_{|\lambda|\to\infty}\frac{\log\Big([(2\im(\lambda)+1)/C_2]^{1/\delta}\Big)}{|\lambda|^\rho}\leq\frac{1}{\delta}\limsup_{|\lambda|\to\infty}\frac{\log|\lambda|}{|\lambda|^\rho}=0.
\end{split}
\end{equation*}
Thus, we have shown that for $\im(\lambda)\geq 0$ and $\rho>1+1/\delta$
\begin{equation*}
\limsup_{|\lambda|\to\infty}\frac{\log|b(\lambda)|}{|\lambda|^\rho}\leq 0.
\end{equation*}
Now, using \eqref{Eq:4}, if we set $\lambda=r\E^{\I\varphi}$, $0\leq\varphi<\varphi+\pi/\rho\leq\pi$, we observe
\begin{equation}\label{MREq:11}
0\leq h_b(\varphi)+h_b(\varphi+\pi/\rho)\leq h_b(\varphi+\pi/\rho)\leq 0,
\end{equation}
where $h_b(\cdot)$ is the indicator function of $b(\lambda)$ with respect to the order $\rho$. Notice that similar calculations can be done for $h_b(\varphi)$ in \eqref{MREq:11} and hence $h_b(\varphi)=h_b(\varphi+\pi/\rho)=0$, for $0\leq\varphi<\varphi+\pi/\rho\leq\pi$.

By hypothesis of the theorem, on one hand we know that $h_{b(t_i)}(\varphi)=0$, for $i=0,1$ and on the other hand, using the evolution formulas for the scattering coefficients \eqref{NLSEq:11}, in particular $b(\lambda,t)=b(\lambda,0)\E^{-2A_-(\lambda)t}$, we obtain for $\lambda=r\E^{\I\varphi}$ ($0<\varphi<\varphi+\pi/\rho<\pi$)
\begin{equation*}
\begin{split}
0&=\limsup_{r\to\infty}\frac{\log|b(r\E^{\I\varphi},1)|}{r^\rho}=\limsup_{r\to\infty}\frac{\log|b(r\E^{\I\varphi},0)|}{r^\rho}\\
&+\limsup_{r\to\infty}\frac{\log|\E^{-2A_-(r\E^{\I\varphi})}|}{r^\rho}=\limsup_{r\to\infty}\frac{\log|\E^{-2A_-(r\E^{\I\varphi})}|}{r^\rho},
\end{split}
\end{equation*}
but the last equality  diverges by \eqref{ThEq:1}. This leads us to a contradiction unless $b(\lambda,0)=0$, and hence by \eqref{NLSEq:11} $b(\lambda,t)=0$ for all $t$. Similar calculations can be done for $\overline{b}(\lambda,t)$, to show that it also vanishes for all time $t$. Finally,  it remains to apply Lemma \ref{Lemma:1} and the theorem is proved.

\end{proof}

Similarly, we have the same results for the KdV case, that is:

\begin{lemma}\label{Lemma:2}
Let $(q,r)$ be a solution for the coupled equations in \eqref{NLSEq:1} and assume that we are in the KdV case, that is $r=-1$. If for some $t_0$ :
\begin{itemize}
\item[i)]   $\int_{\mathbb{R}}(1+|x|)|q(x,t_0)|dx<\infty$ ,
\item[ii)] $|q(x,t_0)|\leq C_1\E^{-C_2 x^{1+\beta}} $, where $C_i,\beta>0$ are constants, $i=1,2$, and $x\geq 0$, and
\item[iii)] the system is reflectionless, i.e., $R_1(\lambda,t_0)=0$,
\end{itemize}
then $q(x,t)=0$ for all $x,t\in\mathbb{R}$.
\end{lemma}

\begin{proof}
The proof is similar to the Lemma \ref{Lemma:1} and we omit it.
\end{proof}

\begin{theorem}\label{theorem:2}
Let $(q,r)$ be a solution for the coupled equations in \eqref{NLSEq:1}.  Suppose  that we are in the KdV case, that is, $r=-1$, and at two different times $t_0<t_1$ we also have the following properties:
\begin{itemize}
\item[i)]   $\int_{\mathbb{R}}(1+|x|)|q(x,t_i)|dx<\infty$ , $i=0,1$,
\item[ii)]  $|q(x,t_j)|\leq C_1\E^{-C_2x^{1+\delta}} $, where $C_i,\delta>0$ are constants,  $i=1,2$, $j=0,1$ and  $x\geq 0$, and
\item[iii)] $\delta$  fulfills
\begin{equation}
\limsup_{|\lambda|\to\infty}\frac{|\re(A_0(\lambda))|}{|\lambda|^\rho}=\infty,\quad\quad \rho>1+1/\delta,
\end{equation}
\end{itemize}
 then $q(x,t)=0$ for all $x,t\in\mathbb{R}$.
\end{theorem}
\begin{proof}
The proof, as it happened with Lemma \ref{Lemma:2}, follows the same pattern given in theorem \ref{theorem:1}.
\end{proof}

As we remarked in the introduction, if we set $A_0(z)=-2\I z^2$ and $r=-q^*$, then the coupled equations \eqref{NLSEq:1} become the nonlinear Schr\"odinger equation. Moreover, to apply theorem \ref{theorem:1} we only need decay on $r$ or $q$ and to fulfill condition $ii)$ of the theorem we need
\begin{equation*}
\limsup_{|\lambda|\to\infty}\frac{|\re(-2\I z^2)|}{|\lambda|^\rho}=\infty,\quad \rho>1+1/\delta.
\end{equation*}

This means that $\rho<2$ and hence we observe that it is enough to have $\delta>1$ for the decay condition. Similarly, if we write $r=-1$ and $A_0(z)=-4\I z^3$, then  the equation \eqref{NLSEq:1} tuns into the Korteweg-de Vries equation and  $\delta>1/2$. A similar argument can be used for  the modified Korteweg-de Vries equation and $\delta>1/2$.  However, we notice that this technique cannot be applied to the first member of the Korteweg-de Vries hierarhy 
\begin{equation}\label{Rem:001}
u_t=u_x.
\end{equation}
This is due to the fact that \eqref{Rem:001} is equivalent to $\eqref{NLSEq:1}$ when  $r=-1$ and $A_0(z)=-\I z$, i.e, the phase velocity is constant and hence all the waves move at same velocity. In addition,  condition $ii)$ of theorem \ref{theorem:2} implies that $\rho>1+1/\delta$, but this is not possible because we also need 
$\limsup_{|\lambda|\to\infty}\frac{|\re(-2\I z^2)|}{|\lambda|^\rho}=\infty$, that is $\rho<1$.

As a final remark, we point out that the decay condition on the theorems \ref{theorem:1} and \ref{theorem:2} can also be on the left hand side of the real line and a similar argument can be applied to adjust the proofs of the previous results.

\medskip

\noindent 
{\bf Acknowledgments.}
I am indebted to  Gerald Teschl for  productive discussions on this topic.

\end{document}